\documentclass[sigconf]{aamas}

\usepackage{soul}
\usepackage{url}
\usepackage{amsmath}
\usepackage{amsthm}
\usepackage{amsfonts}

\usepackage{nicefrac}

\usepackage{cleveref}
\crefname{algocf}{algorithm}{algorithms}
\Crefname{algocf}{Algorithm}{Algorithms}

\usepackage[ruled,linesnumbered,commentsnumbered,noline]{algorithm2e}
\DontPrintSemicolon
\SetNoFillComment
\usepackage{float}
\usepackage{bm}

\newcommand{\N}{\ensuremath{\mathcal{N}}} 
\NewDocumentCommand{\apr}{o o}{%
 \ensuremath{\alpha\IfValueT{#1}{_{#1}}\IfValueT{#2}{(#2)}%
}}
\NewDocumentCommand{\prd}{o o}{%
 \ensuremath{b\IfValueT{#1}{_{#1}}\IfValueT{#2}{(#2)}%
}} 
\newcommand{\nullagent}{\ensuremath{\emptyset}} 
\newcommand{\bigbag}{\ensuremath{\mathcal{B}}} 
\newcommand{\E}{\ensuremath{\mathbb{E}}} 

\usepackage{balance} 

\newtheorem{theorem}{Theorem}
\newtheorem{lemma}{Lemma}
\newtheorem{corollary}{Corollary}

\setcopyright{rightsretained}
 \acmConference[GAIW'26]{Appears at the 8th Games, Agents, and Incentives Workshop (GAIW-26). Held as part of the Workshops at the 25th International Conference on Autonomous Agents and Multiagent Systems.}{May 2026}{Paphos, Cyprus}{Armstrong, Curry, Hosseini, Mattei, Tsang, Wąs (Chairs)}

\title[\textsc{PeerBTS}: Incentivizing Effort in Strategyproof Peer Selection]{\textsc{PeerBTS}: Incentivizing Effort in Strategyproof Peer Selection}

\author{Harper Lyon}
\affiliation{
  \institution{Tulane University}
  \city{New Orleans}
  \country{United States of America}}
\email{hlyon@tulane.edu}

\author{Omer Lev}
\affiliation{
  \institution{Ben-Gurion University}
  \city{Beersheba}
  \country{Israel}}
\email{omerlev@bgu.ac.il}

\author{Nicholas Mattei}
\affiliation{
  \institution{Tulane University}
  \city{New Orleans}
  \country{United States of America}}
\email{nsmattei@tulane.edu}

\begin{abstract}
 Peer selection, the evaluation and selection of agents by their peers, is an important problem in the field of computational social choice; with applications to grading in massively online courses (MOOCs) and academic peer review. Current existing algorithmic and empirical work focuses on developing and analyzing novel \emph{strategyproof} mechanisms, wherein no agent has an incentive to misreport their evaluations. However, the majority of published mechanisms share a flaw: they do not \emph{reward} agents for any effort expended during the evaluation process. In cases where high quality evaluations are costly to produce this missing incentive fails to align agents with an overall goal of accurate selection. To address this gap we first prove theoretically that incentivizing effort in peer selection requires information beyond a single evaluation. We then propose \textsc{PeerBTS}, a mechanism that combines a peer-prediction lottery, leveraging work on the Robust Bayesian Truth Serum, with any existing peer-selection mechanism to incentivize effort while remaining Bayes-Nash incentive compatible. We find that while an incentive-compatible peer-selection mechanism using agent predictions to incentivize effort is possible it requires adjustments to the assumed problem context and limits other mechanistics properties. We additionally present a series of non-strategic simulations to validate incentives and evaluate the performance of PeerBTS relative to existing strategyproof peer selection mechanisms. Finally, we discuss the results of an initial study on the validity of peer-prediction from a small academic workshop.
 \end{abstract}


\newcommand{\BibTeX}{\rm B\kern-.05em{\sc i\kern-.025em b}\kern-.08em\TeX}

\begin{document}


\pagestyle{fancy}
\fancyhead{}


\maketitle 


\section{Introduction}

Peer evaluation -- the process by which a group evaluates its own members and chooses the best of them -- has been used by people for millennia. The various ``mechanisms'' people used ranged from hoping for divine interventions and lotteries\footnote{From the bible's Numbers chapter 17 and Joshua chapter 7 to Delphic oracles studying goat entrails and the Athenian \emph{Boule} \cite{MB24} and the Venetian Doge selection mechanism \cite{MG07}, they all tried to improve their selection process with lotteries.}, to attempts at forms of elections \cite{merrifield2009telescope}. However, the salience of this problem became much more significant as technology enabled the groups conducting peer evaluation to grow dramatically, e.g., Massive Open Online Courses (MOOCs) like Coursera. Needing a way to grade the tasks done by huge classes (up to several million), practitioners have turned to students to rate each other \cite{piech2013tuned,luo2014peer,caragiannis2015aggregating}. Similarly, Computer Science conferences face an increasing demand for reviewers, and many have resorted to requiring those submitting papers to also review papers. 

As the relevance of peer evaluation grew, the issues with existing systems became clearer. First and foremost, they are not strategyproof, and thus may incentivize agents to misreport their views, as it may improve their chances of getting chosen themselves \cite{DBLP:journals/corr/ArdabiliL13a,gurevych2024reviewer}. Consider, for example, the system used in many CS conferences, in which a submission's ``grade'' is an average of its reviewers' grades, and acceptance is for the top $k$, for some value of $k$. All reviewers will be better off giving $0$ as their grade to anyone, as while they cannot affect the grade their submission receives, they can lower the threshold grade to be in the top-$k$, so their work might be selected. This is typically called ``torpedo reviewing'' in the literature and has been documented at conferences \cite{DBLP:conf/nips/JecmenZLSCF20}.

Several strategyproof mechanisms have been suggested for peer evaluation, in which no agent can improve their chance of selection by changing their review on others \cite{DBLP:conf/tark/AlonFPT11,crediblesubset15,matteiPS2019,lev2023peernomination} (more coverage is available in several overviews \cite{lev2024impartial,DBLP:journals/corr/abs-2210-01984}). The problem with strategyproof mechanisms is that, by definition, they force a separation between the evaluations an agent gives on others and the possibility of that same agent being selected. This separation is desirable in the mechanism design literature as the assumption is that removing this connection will cause agents to report truthfully \cite{faliszewski2010ai}.

However, when there is a cost -- whether of time, effort, or monetary value (as with reviewing papers, grant proposals, or Coursera student work) -- the disconnect between evaluations and agents' personal benefit is detrimental to the quality of the peer evaluation since there is no \emph{within mechanism incentive to exert effort}. Previous attempts at tying the evaluations to rewards for the agent providing reviews, e.g., \citet{merrifield2009telescope,srinivasan2021auctions}, use things like rewarding agents for giving a review close to the average of others agents' reviews are inherently not strategyproof, as agents are being rewarded \emph{not for giving their truthful views, but for being able to predict consensus opinion}. Consider a reviewer for a CS conference using such a reward scheme, who finds a fundamental problem with a paper's proof. If the reviewer thinks other reviewers won't pick up on the proof's problem, they are better-off ignoring it and giving the proposal the same review as the others, as that will increase their own chances of selection.

\paragraph{Contribution.} We propose \textsc{PeerBTS}, a solution using a Bayesian Truth Serum (BTS) mechanism \cite{prelec2004}, which we pair with \textsc{PeerNomination} \cite{lev2023peernomination}, a known strategyproof peer evaluation mechanism.\footnote{We will see our method can incorporate any peer evaluation mechanism, however for clarity we focus on \textsc{PeerNomination}.} We show that existing strategyproof mechanisms do not incentivize effort and that, if the only data given by agents is their review of others, it is \emph{impossible} to create a mechanism that incentivizes truth-telling while also incentivizing agents to put in effort to evaluate others. We give a general framework to design reward mechanisms as a combination of scoring and selection mechanisms, prove when and which properties can be retained; and instantiate and evaluate one such mechanism -- \textsc{PeerBTS} -- that is Bayes-Nash incentive compatible.

\section{Related Work}

Peer selection has a rich history, but as an object of academic study has typically been modeled in the context of academic peer review. 
Peer review as done by independent panels has been an object of much study in the past \cite{CCS81a,MEFF90a,WeWo97a}, but in the past few decades there has been a move for a broad base of reviewers, leading many organizations to move beyond expert panels and adopt true peer selection by requiring applicants to also serve as reviewers. Moving towards community review and away from hand selected expert panels puts new stresses on the peer review process, including issues of incentive and dishonesty, which is where the tools of mechanism design have proved useful. Much of formal work was first inspired by \citet{merrifield2009telescope}, who proposed a basic Borda mechanism to adjudicate the allocation of time on telescopes. A variety of more sophisticated mechanisms have been proposed, including \textsc{CredibleSubset} \cite{crediblesubset15} and \textsc{ExactDollarPartition} \cite{matteiPS2019} (see \citet{DBLP:journals/corr/abs-2210-01984} for a recent survey). This research focuses primarily on strategyproofness, though there are further concerns, such as group fairness \cite{DBLP:conf/nips/0001M023a} or handling conflict of interest \cite{DBLP:journals/corr/abs-1806-06266,DBLP:conf/nips/JecmenZLSCF20}. Our work is similar, as we are concerned with an property other than strategyproofness, namely incentivization \emph{within} the mechanism, which has not been studied.

Beyond the traditional research in peer selection we also pull from the literature on group predictions/forecasting. Mechanisms for high quality group predictions are not a new idea and have been explored for decades \cite{degroot1974reaching,genest1986combining}. There are a number of traditions in this space, such as the use of markets for predictions \cite{wolfers2004prediction,arrow2008promise}, but we draw from 
a related tradition of survey scoring, where the goal is to solicit honest responses from participants. The leading work in this space is the Bayesian Truth Serum \cite{prelec2004}, which asks agents to to make a prediction about the population response rate. Further work on group prediction -- most notably the Good Judgment Project \cite{tetlock2016superforecasting} and related efforts -- has produced a wealth of expansions on this idea, including the Robust Bayesian Truth Serum \cite{robustBTS,radanovic2013robust} which we use as a core component of \textsc{PeerBTS}. Research in this area has proven that the techniques of mechanism design can be applied profitably to group decision making, and so it is a natural choice for our extension of traditional peer selection mechanisms.

A related, though distinct, line of research focuses on incentivizing effort/accuracy through more active means. In particular \citet{precisioninpeergrading2024} draw on work on Massive Open Online Courses (MOOCs) \cite{tunedpeerassement2013} to present a mechanism which incentives effort from self-interested participants. While our goals are similar our approaches are quite different. In particular, Chakraborty et al. use known-value probe papers to directly estimate agent effort and reward agents with bonuses exterior to the selection process itself. This approach is well designed and could be quite profitably implemented in cases where the center has some access to ground truth, making it a good fit for MOOCs. Our own mechanism is more setting agnostic, as we do not rely on having any access to the ground truth nor on exterior rewards, though we do accept several tradeoffs for this flexibility.


\section{Preliminaries and Notation}

Our running example is a conference where agents submit proposals, which they wish to see selected, and also serve as reviewers for the conference \cite{DBLP:journals/cacm/Shah22}. For ease of exposition, we suppose that agents are providing \emph{reviews} on a set of \emph{proposals} that are under review. The review is the act of submitting their belief about a proposal's quality, and we will also ask agents for a prediction of the reviews of other agents as part of the mechanism. We will also refer to the \emph{review board} as the set of agents reviewing a particular proposal, and the review assignment of a single agent as the set of proposals they are to review. Unlike previous peer evaluation research, we assume there is some cost (effort, time) for agents to generate their reviews, and an agent can choose to pay this cost and generate a quality review, or not and generate a low quality one. Additionally, taking a standard mechanism design approach, agents will only do what they are \emph{incentivized} to do \cite{borgers2015introduction}. Note that we are not disparaging a typical conference reviewer who may have intrinsic or altruistic motives for submitting good reports, rather we are working from first principles in incentive and mechanism design.

Following the broader literature, we define a peer selection problem over a set of agents $\N$, with $|\N| = n$. Each agent $i \in \N$ has a private evaluation function $\sigma_i: \N \rightarrow \mathbb{R}$, representing their perception (review) of the quality of all agents involved which can be viewed as a (strict) ordinal ranking function $\N \rightarrow \{0,1,\ldots,n\}$ by ordering the agents by quality (breaking ties in some manner). Taken together, we refer to the set of all evaluating functions of all agents as a \emph{profile}, $\sigma$, and the set of all possible profiles as $\Sigma$.

In peer selection there is a specified number $k$ of selections or winners, $\mathcal{W}$, where $\mathcal{W} \subseteq \N$ and, ideally, $|\mathcal{W}| = k$. Informally, we want to select the ``best'' $k$ sized subset of proposals. To do so, we define a selection mechanism as a function from profiles to a selected set, formally $f: \Sigma \rightarrow 2^{\N}.$ We call any such function $f$ \emph{exact} if $\forall \sigma \in \Sigma, \ |f(\sigma)| = k $.

In real life instances of peer selection we typically do not solicit a complete evaluation or ranking of all involved agents, as this would be far too onerous. Instead we define a review assignment function $A: \mathcal{N} \rightarrow 2^{\mathcal{N}}$ such that $\forall i \in \mathcal{N}, i \notin A(i)$. This assigns each agent to a review assignment $A(i)$, and we can also refer to the set of reviewers, review board, assigned to a specific proposal as $A^{-1}(i)$. We further restrict the review assignment functions to be $m$-regular, where $\forall i \in \mathcal{N},\, |A(i)| = m$. In this way $m$ parametrizes the number of proposals that any agent is expected to review, and $m$ tends to be relatively small compared to $n$ as a matter of practicality.

Unlike previous peer evaluation research, we assume there is some cost to agent $i$ to obtain $\sigma_{i}$, akin to putting in time to read a paper for a conference review. The agent can chose not to put in this effort, resulting in low-quality or random data going into the mechanism. Hence our central question is \emph{can we design mechanism incentives to encourage agents to put in the effort to generate high quality reviews without destroying the incentives to be truthful?}

\subsection{Desirable Properties}
Ideally, peer selection mechanism should satisfy \textbf{strategyproofness}, under the assumption that the only thing agents desire is to have their proposal selected. Informally, this means that no agent involved in the selection process has an incentive to misreport their preferences to increase their chance of selection; hence, we can ideally expect honest reports from every reviewer to a strategyproof mechanism. Formally, a mechanism $f$ is strategyproof if
\small
\begin{align*}
\forall i \in N, \forall \sigma_i,\quad & i \notin f((\sigma_1, \ldots, \sigma^*_i, \ldots, \sigma_n)) \\
\Rightarrow\quad & i \notin f((\sigma_1, \ldots, \sigma_i, \ldots, \sigma_n))
\end{align*}
\normalsize
where $\sigma_i^*$ is agent $i$'s true private ranking function. In probabilistic mechanisms we require that $i$ cannot increase their probability of being selected by misreporting.

A related, weaker, property is \textbf{incentive compatibility}, where truthfulness is a Nash equilibrium, i.e., when all are truthful it is never beneficial for any agent to deviate,
\small
\begin{align*}
\forall i \in \N, \forall \sigma_i,\quad & P(i \in f((\sigma_1^*, \ldots, \sigma_i^*, \ldots, \sigma_n^*))) \\
\geq\quad & P(i \in f((\sigma_1^*, \ldots, \sigma_i, \ldots, \sigma_n^*)))
\end{align*}
\normalsize
An additional desirable property is that no particular agent have more influence than others -- \textbf{anonymity}.


\subsection{\textsc{PeerNomination}}

While \textsc{PeerBTS} is a general mechanism that can be combined with any strategyproof (or even incentive compatible) mechanism, we instantiate and evaluate it with \textsc{PeerNomination} \cite{lev2023peernomination}. Broadly, \textsc{PeerNomination} is approval based -- agents either approve others or not, and do not rank them or give them more granular grades. Agents approve the top $\lfloor \nicefrac{k}{n} \cdot m\rfloor$, i.e., the rounded down share that needs to be selected from the population in their review assignment. The agent in the next position, rank $\lfloor \nicefrac{k}{n} \cdot m\rfloor+1$, is considered partially approved -- given $(\nicefrac{k}{n} \cdot m - \lfloor \nicefrac{k}{n} \cdot m\rfloor)$ points. The set of accepted agents are those who have been approved with $\nicefrac{m}{2}$ points, i.e., approved by most of their reviewers.

Note that \textsc{PeerNomination} is not an exact mechanism, as it does not guarantee that exactly $k$ agents are chosen. The mechanism has further variants, allowing for a weighing of the reviews of different agents, but that is outside the scope of this paper. As \textsc{PeerNomination} uses ranking information for its outcome, we shall assume from now on that $\sigma_{i}(j)$ returns the rank of agent $j$ by agent $i$.

\section{Existing Mechanisms Lack Incentives}

%

Our first key result is that mechanisms in which agents are simply asked to review other agents (via grade, ranking, approval, etc.) can never be both strategyproof (or even the weaker incentive compatible) and incentivize agents to put in the effort that results in good reviews.

\begin{theorem}
A peer-evaluation mechanism cannot be anonymous and strategyproof (or incentive-compatible) and incentivize agents to produce a high quality (i.e., truthful) $\sigma$, if its input is only $\sigma$.
\end{theorem}
\begin{proof}
Since the mechanism is anonymous, there is a profile $\bar{\sigma}\in \Sigma$ for which agent $i$'s review will influence the outcome. For any such profile, because the mechanism is strategyproof (incentive-compatible), and agent $i$'s evaluation is private and does not depend on others, then agent $i$'s probability of being selected is not influenced by their review. So whether $i$ reports some randomized $\sigma'_{i}$ or an exact $\sigma_{i}^{*}$, its chances of being selected remain the same. Since the only thing the agent cares about is being selected, and we assume there is a cost to finding $\sigma_{i}^{*}$, there is no incentive to report $\sigma_{i}^{*}$, and the mechanism has no way of incentivizing agents.
\end{proof}

%
%

\begin{corollary}\label{2bitsCor}
A peer-evaluation mechanism that is anonymous, strategyproof (or incentive-compatible) and incentivizes agents to produce a high quality $\sigma$, requires agents to provide at least one more piece of information besides $\sigma_{i}$.
\end{corollary}

Throughout the paper we assume that providing a high quality review $\sigma\in\Sigma$, requires some sort of effort, e.g., reading a paper to give a review on it, and that if this effort is made, it will result in a truthful and high quality $\sigma^{*}$. This holds even if that effort was motivated by a different incentive. For example, if a mechanism gives a bonus to an agent able to estimate what the average reviewer will say of a proposal (as is the case in the NSF pilot study \cite{merrifield2009telescope}), the agent will need to read the proposal to give an answer, but by reading it, they are now able to write a high quality review, even though they have not directly been rewarded for their review. Therefore any incentive towards effort will have a side effect of improving review quality.


\section{\textsc{PeerBTS}: Peer Selection with Bayesian Truth Serum}

As \Cref{2bitsCor} points out, we must add another piece of information to the mechanism in order to have one that incentivizes both truth-telling and high-quality reviews. To do so, we combine a Bayesian Truth Serum mechanism with an existing strategyproof peer-evaluation algorithm. We used \textsc{PeerNomination} here, in order to present an explicit implementation, but one could use others.

\subsection{Mechanism Overview}

\begin{figure*}[t]
 \centering
 \includegraphics[width=0.6\textwidth]{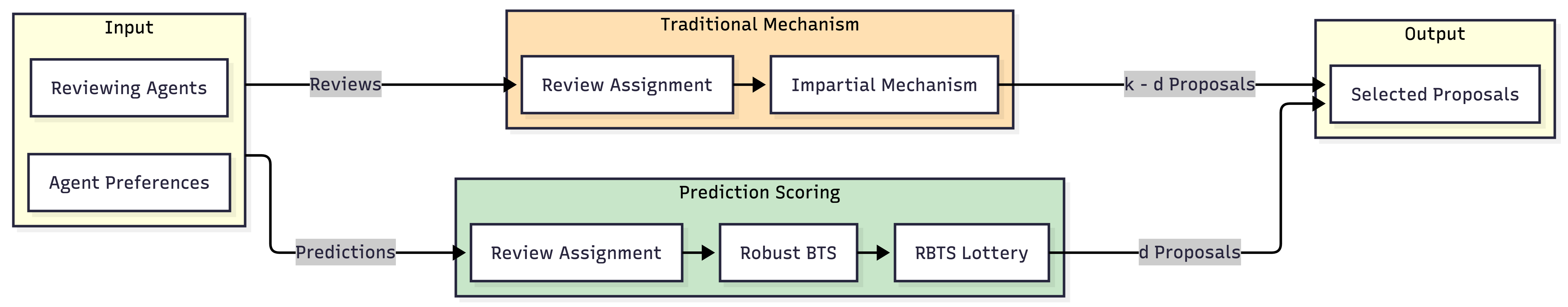}
 \caption{Overview of information flow in the PeerBTS mechanism.}
 \label{fig:mechanism_flow}
 \Description{Flowchart representing how reviews and predictions are split between mechanisms in PeerBTS.}
 \vspace{-10pt}
\end{figure*}

\textsc{PeerBTS} works at a high level by taking any existing strategyproof mechanism $f(\sigma)$ and a reward mechanism $h(\sigma,b)$ which relies on $\sigma$ and another piece of information $b$. We build \textsc{PeerBTS} based on the Bayesian Truth Serum (BTS) family of mechanisms which ask agents for their opinions and their estimates of the average of others' opinions, rewarding agents based the accuracy of their predictions.

\textsc{PeerBTS} first constructs $S_{f}$, the set of selected agents using $f(\sigma$), with the number of selected being $k-d$, where $d<k$ is some integer parameter. Then, using $h(\sigma,b)$, it creates a set $S_{h}$ of $d$ agents rewarded by $h$ (see \Cref{fig:mechanism_flow}). The set of winners is $S_{f}\cup S_{h}$, notice that even if $f$ was an exact mechanism, the combined set may be smaller than $k$, as sometimes $S_{f}\cap S_{h}\neq\emptyset$.

\subsection{Reward Mechanism: Predict and Evaluate}

In \textsc{PeerBTS} we use approval functions and predictions of approval rates as the additional information required by \Cref{2bitsCor}. We first define the necessary part for \textsc{PeerNomination}: an approval function, indicating which agents rank proposal $i$ at rank $q$ or above: $\apr[i]: A(i) \rightarrow \{0, 1\}.$ Typically this is defined over a quota $q \in (0, m)$ such that 
\small
\[
\apr[i][j] = \begin{cases}
 1 & \text{if } \sigma_j(i) \leq q\\
 0 & \text{otherwise}
\end{cases}
\]
\normalsize
with variations allowing for probabilistic approval when $q$ is a non-integer. We denote the set of all such approval functions as the approval profile $\apr$. 

We also collect reviewers' predictions of approval rates. We define the belief function $\prd[i]: A(i) \rightarrow [0,1]$ such that $\prd[i][j]$ represents the proportion of reviewers $j$ believes will approve proposal $i$. Formally, $\prd[i][j] = \mathbb{E}_j\left[\frac{1}{m} \sum_{s \in A^{-1}(i)} \apr[s][j]\right] \forall j \in A(i)$. We call the collection of belief functions a belief profile $\prd = \{\prd[1], \prd[2], \ldots, \prd[n]\}$.

\subsection{Reward Mechanism: Prediction Scoring}

To reward agents for their predictions we define a score function, which assigns a numerical score to each agent according to the reviews and predictions. Formally:
$h: \Sigma \times \beta \rightarrow \mathbb{R}^{|\mathcal{N}|}.$

We use the Robust Bayesian Truth Serum (RBTS) \cite{robustBTS}, a variant of the well known Bayesian Truth Serum \cite{prelec2004}, to score agent reviews and predictions. RBTS is preferable in this context as it is incentive compatible even when dealing with small populations. RBTS takes as input the set of approval functions $\apr$ and belief functions $b$ and assigns each reviewer a score for each response and prediction within their review assignment. Considering the reviews of a single agent $i$, each agent $j$ in the review board $A^{-1}(i)$ is assigned a reference agent $k = j + 1 \mod |A^{-1}(i)|$ and peer agent $l =j + 2 \mod |A^{-1}(i)|$, both their review and prediction are scored against $j$ and $k$. \Cref{alg:rbts} gives our formal implementation.

\begin{algorithm}[t]
\small
\SetAlgoLined
\caption{Robust Bayesian Truth Serum}
\label{alg:rbts}

\KwIn{Agent $i \in N$, $m$-regular review assignment $A$, approval profile $\apr$, belief profile $\prd$, proper scoring rule $R$}
\KwOut{Score vector $<s_j \in [0,2] | \forall j \in A^{-1}(i)>$}

Scores $:= <0 \forall j \in A^{-1}(i)>$ \;
\tcp*[l]{Score agents in review board}
\For{$j \in A^{-1}(i)$}{ 
 $k := j + 1 \mod m$\;
 $l := j + 2 \mod m$\;
 $\delta := \min(\prd[k][i], 1 - \prd[k][i])$\;
 \uIf{$\apr[i][j] = 1$}{
 $w_j := \prd[i][k] + \delta$ \;}
 \uElseIf{$\apr[i][k] = 0$}{
 $w_j := \prd[i][k] - \delta$ \;}
 Scores$(j) = R(w_j, \apr[i][l]) + R(\prd[i][j], \apr[i][l])$
}
\Return Scores
\vspace{-4pt}
\end{algorithm}
\normalsize
For our purposes what is most important is that RBTS is Bayes-Nash incentive compatible for any population of $3$ or greater \cite[Theorem 9]{robustBTS}, so as long as every agent $a \in \N$ has a review board $|A(a)| \geq 3$, we can safely use RBTS as a scoring mechanism. 

\subsection{Combining Reward \& Quality Selections}

Combining the output of the peer selection mechanism and the Robust-BTS score in a way that maintains incentive compatibility is not trivial: There are subtle issues in many of the obvious ways in which one might combine these two functions. The most straightforward approach is to use $h(\sigma, b)$ to select the top $d$ agents from $\N \backslash f(\sigma)$, to avoid double selection and guarantee exactness. This. however, is not strategyproof nor incentive compatible. As a counter example, consider a 2-selection setting with four agents ($1, 2, 3, 4$), strategyproof selection mechanism $f$, and scoring function $h$. On truthful reviews from all agents we get profile $\sigma$ and similarly on effortful predictions we get $b$. Suppose that these two inputs give us the rankings: $f(\sigma) = [1, 2, 3, 4]$ and $h(\sigma, b) = [2, 3, 1, 4]$. Following our simple combination we take the top agent from $1 \in f(\sigma)$ and then the top agent of $h(\sigma, b)\setminus \{1\}$. On the truthful rankings this would give us the set ${1, 2}$, the two top agents of their respective rankings.

Now say that by misreporting, agent $3$ can produce the profile $\sigma'$ and ranking $f(\sigma') = [2, 1, 3, 4]$. This is not a violation of $f$'s strategyproofness since in this ranking $3$ is indifferent between $f(\sigma)$ and $f(\sigma')$. However, this results in a different outcome as the rankings are combined: under $f(\sigma') = [2, 1, 3, 4]$ and $h(\sigma', b) = [2, 3, 1, 4]$ we would first select $2 \in f(\sigma)$, and then select $3 \in h(\sigma, b) \setminus\{2\} = [3, 1]$. By promoting an agent in one ranking, they can be removed from consideration in the other, breaking strategyproofness. Similar counter examples can be constructed in a probabilistic context as agents are incentivized to shift probability mass to higher ranked agents to improve their odds of selection.

To achieve Bayes-Nash incentive compatible (BNIC) combination we accept a measure of inexactness in the final selected set by allowing agents to be selected twice. This can be done deterministically, but a randomized approach is superior as it ensures any increase in prediction score increases selection probability. We treat the prediction scores as stakes in a joint lottery, inspired by \citet{matteiPS2019}, and draw $d$ participants. This gives us a complete mechanism in \Cref{alg:lottery-mechanism}.

\begin{algorithm}[t]
\small
\SetAlgoLined
\caption{Lottery Formalization}
\label{alg:lottery-mechanism}

\KwIn{Agents $\N$, $m$-regular review assignment $A$, approval function $\apr$, prediction function $\prd$, target selection size $d$, weighting exponent $\epsilon$}
\KwOut{Selected set $S \subset \N \bigcup \{\nullagent\}$}

$\mathcal{B} := []$ \tcp*{Initialize list to hold winners for each sublottery}
\For{$i \in \N$}{
 \tcp*[l]{Collect reviews of $i$}
 Actual := $< \apr[i][j] | j \in A^{-1}(i) >$ \;
 \tcp*[l]{Collect predictions}
 Predicted := $< \prd[i][j] | j \in A^{-1}(i) >$ \;
 \tcp*{Score using \Cref{alg:rbts}}
 $\text{Scores}_i$ := \texttt{RBTS}(i, Actual, Predicted) \;
 \For{$j \in A^{-1}(i)$}{
 $L_i(j) := \frac{\text{Scores}(j)^\epsilon}{2^\epsilon \cdot m}$\;
 }
 $L_i(\nullagent) := 1 - \sum_{j \in A^{-1}(i)} L_i(j)$\;
 $s_i \sim L_i$\;
 $\mathcal{B} \leftarrow s_i$
}

\tcp*[l]{Draw $d$ agents from $B$ without replacement to create winning set} $W \sim \binom{B}{d} \setminus \{\emptyset\}$\;
\Return{$W$}
\vspace{-4pt}
\end{algorithm}

\normalsize
\subsection{Conditions for Incentive Compatibility}


Due to space constraints we defer to \Cref{apx:IC_Proof} the proofs of the RBTS mechanism requirements for an admissible prior.

\begin{theorem}
\label{thm:bnic_proof}
If agent preferences are drawn from a distribution in $T$ and $m \geq 3$ then the RBTS lottery (\Cref{alg:lottery-mechanism}) is Bayes-Nash incentive compatible.
\end{theorem}
 
 These conditions are met both where there is a common knowledge of how beliefs are formed by all agents, which is a very strong condition. But they are also met when all agents derive their preferences from the same source/process, including that agents have the same initial chance of developing any particular belief, and that any given belief has at least some chance of being developed \cite{regenwetter2006behavioral}.

\section{Theoretical Properties of \textsc{PeerBTS}}
In this section we give an intuition for \textsc{PeerBTS} guarantees.

\subsection{Agent Incentives}

From \Cref{thm:bnic_proof} agents have an incentive to maximize the score returned by \Cref{alg:rbts}. We know that there is an equilibrium where all agents review honestly, but it is useful to examine this more directly. Specifically, we are interested in any incentive that pushes an agent to report their reviews honestly. A detailed model for agent utility or reliability is beyond the scope of this paper, but even a simple model of agent incentives is useful for understanding the functionality of \textsc{PeerBTS}.

\subsubsection{Incentive Example}
To build intuition we present a simple example calculation of RBTS scores to demonstrate the incentives. We assume for the following that the honest rate of approval for agents is $.2$, representing an uncommon but not highly unlikely review. To simplify the calculations, we also assume that all agents other than one -- agent $j$ -- predict randomly, drawing their prediction uniformly from $[0, 1]$. This gives us consistent expected values for most of the score calculation, $\E(\prd[i][k], \E(\delta) = .5$ which gives us $\E(w_j) = \apr[i][j]$ by the definition of $w_j$.

We first compute agent $j$'s expected information score $\E(R_q(w_j, \apr[i][k])$ if they report randomly (for simplicity reporting $0$ or $1$ with equal probability). Splitting the expected score into these two cases we get $.5(\E(R_q(0, \apr[i][k])) + \E(R_q(0, \apr[i][k]))$. We also split the expected value for $\apr[k][i]$, giving us 
\small
$$.5(.2(R_q(0, 1) + R_q(0, 1)) + .8(R_q(0, 0) + R_q(1, 0))$$
\normalsize
with a final score of $.5(.2(0 + 1) + .8(1 + 0)) = .5.$
If, instead, $j$ reports honestly according to $p$ we would instead see a final expected score of
\small{$$.2^2(R_q(0, 1) + R_q(0, 1)) + .8^2(R_q(0, 0) + R_q(1, 0)) = .68$$}\normalsize
meaning that in this case $j$ has an incentive to report honestly. 

When considering the prediction score the incentive is even easier to see. If $j$ predicts randomly with $\E(\prd[i][j]) = .5$ then the expected prediction score $\E(R_q(\prd[i][j], \apr[i][k])$ does not even need to be broken up based on $\E(\apr[i][k])$, as $R_q(.5, 0) = R_q(.5, 1) = .75$. In all cases a random prediction has an expected value of $.75$, so any score higher than this is preferable for agent $j$. If $j$ predicts $\prd[i][j] = .2$, we have 
\small{$$\E(R_q(\prd[i][j], \apr[i][k]) = .2R_q(.2, 1) + .8R_q(.2, 0) = .84$$}\normalsize
 so in this case we also see a clear incentive for agent $j$ to improve their predictions.

\subsubsection{Formal Analysis}

As shown in \Cref{thm:bnic_proof} under ideal conditions (in the honest equilibrium) the RBTS lottery incentivizes honest reporting and accurate predictions. Our example calculations show that these incentives can exist in at least one case when other agents are not playing the equilibrium strategy but the question remains: what are the limits of the incentives for honesty and accuracy when using the RBTS lottery? Assuming the agent population reports according to some distribution $p^*$ (even if this is a composite distribution) each agent has the following general incentives:

\paragraph{Reports.}
Agents are incentivized to report according to $p^*$. This can be shown by careful application of a result from the original description of RBTS \cite[Lemma 8]{robustBTS}, but this is not necessary. Instead note that if all other agents are reporting according to $p^*$, then their behavior is equivalent to that of honest agents in world where the true report distribution $p = p^*$. We know that honest agents form an equilibrium around truthfully reporting according to $p$, the same is true for dishonest agents all reporting $p^*$, so every agent has an incentive to report according to $p^*$.

\paragraph{Predictions.}
Agents are incentivized to accurately predict $p^*$, a well known property of quadratic scoring rules \cite{selten1998axiomatic} and a result proved directly by the creators of the RBTS \cite[Lemma 7]{robustBTS}. Hence, the incentives in the RBTS (and therefore our lottery) are fundamentally incentives to follow the pack. In the case of honest agents reporting from a shared prior over observations, this incentivizes honesty. In other cases it incentivizes understanding and following the behavior of the broader agent population.

While obviously an incentive towards honesty is ideal, the pack-following nature of the RBTS can still be useful. Recall that the primary goal of \textsc{PeerBTS} is to provide agents an incentive to expend some effort in the reviewing process. While we do not model in detail the application of effort -- as any such model would be highly ad-hoc -- it is clear that an agent $j$ acting with no effort whatsoever will need to either report a constant approval value ($0$ or $1$) or consistently randomize their report. Since all agents have an incentive to follow the reporting of other agents as long as $p_j \neq p^*$ agent $j$ has an incentive to expend effort to move closer to $p^*$ -- so even in cases where \textsc{PeerBTS} fails to be BNIC it still incentivizes effort. This is equivalent to a Keynesian guessing game \cite{keynes1936longterm}, and should exhibit similar dynamics.

\section{Empirical Evaluation of \textsc{PeerBTS}} \label{sec:empirical}

Following the experimental methods of previous work on strategyproof peer selection \cite{matteiPS2019} we evaluate \textsc{PeerBTS} in a series of simulations.\footnote{All experimental code and data -- including \textsc{PeerBTS} implementation -- will be published on GitHub in the event of publication.}

\subsection{Experimental Setup} 

\begin{figure}[t]
 \centering
 \includegraphics[width=0.6\linewidth]{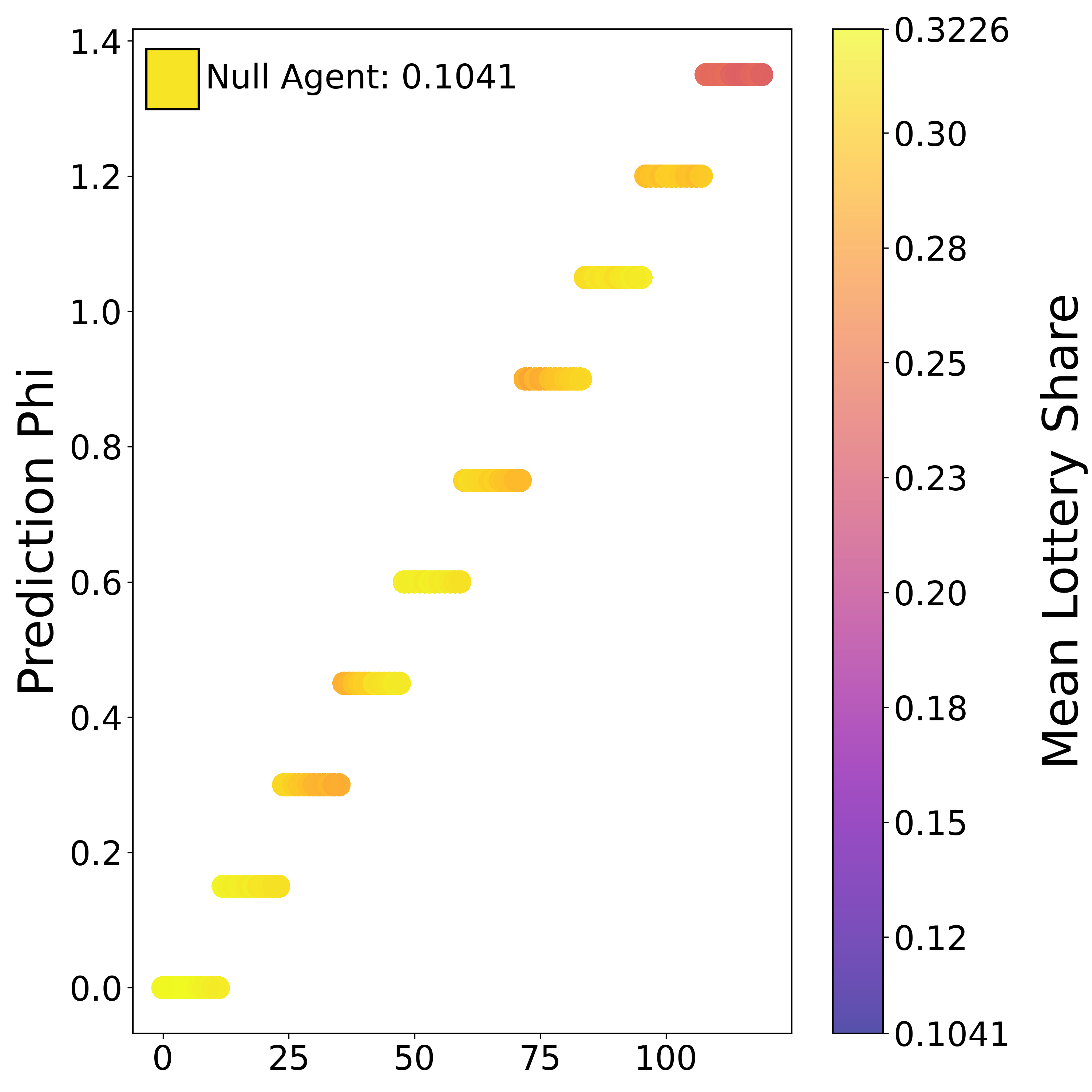}
 \caption{Mean lottery share for agent populations with Deciles treatment in simulation with $m = 3, e = 8.0$, and $\varphi = 0.5$.}
 \label{fig:lottery_share_graph}
 \Description{Graph representing the average lottery share received by agents separated by prediction accuracy.}
 \vspace{-10pt}
\end{figure}

We extend the framework of \citet{lev2023peernomination} and only briefly describe the pre-existing methods here; a more complete specification of certain details can be found in that paper. Complete parameters used can be found in \Cref{apx:params}.

\subsubsection{Noisy Profile Generation}

Following the Condorcet perspective on social choice we assume that there exists an underlying true ranking of agent quality. For simplicity we assume that the true ranking of an agent is identical to their number, so that agent $1$ is the highest quality agent, agent $2$ the second, and so on. To model potential inaccuracies in agent reviews we construct our profiles from noisy observations of this ranking according to a Mallows model \cite{mallows1957non}. We sample permutations of a base ranking $\phi$ based on a dispersion parameter $\varphi \in [0,1]$, creating a distribution where the probability of an alternative ranking $\phi'$ is $\pi_{R, \varphi} \propto \varphi^{KT(\phi, \phi')}$, where $KT(R, R')$ is the Kendall-Tau distance between $\phi$ to $\phi'$ \cite{kendall1938new}. 


\begin{figure*}[t]
 \centering
 \includegraphics[width=0.6\textwidth]{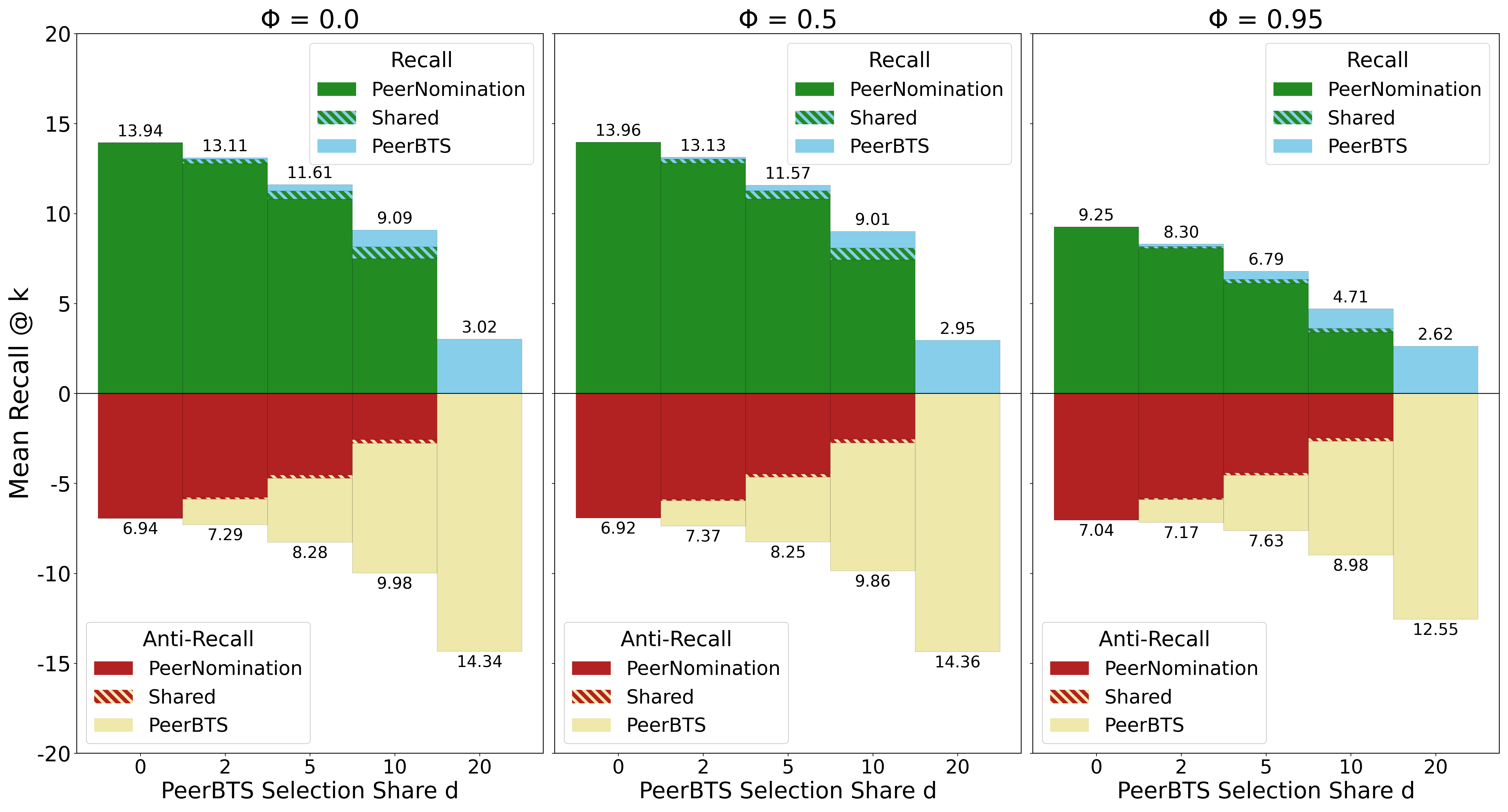}
 \caption{Mean recall and anti-recall of ground truth for combined \textsc{PeerNomination} and \textsc{PeerBTS} with Decile prediction model.}
 \label{fig:recall_graph}
 \Description{Graph presenting recall and anti-recall of PeerBTS at different levels of hybridization and overall noise presence.}
 \vspace{-10pt}
\end{figure*}


Also note that as a continuous, shared distribution the Mallows model satisfies our requirements for an admissible prior as shown in \Cref{lem:distribution} and \Cref{lem:state_priors}. For this reason we do not simulate strategic behavior among agents, though this would be an interesting direction for future research.

\subsubsection{Prediction Generation}

\textsc{PeerBTS} solicits predictions from the entire agent population, and as such the generation of agent predictions is a key part of our simulation. Following our profile sampling methodology, we treat agent predictions as (potentially) noisy observations of a true profile $\sigma$ which represents the preferences of all agents. Again mirroring our profile sampling methodology we use a Mallow's model to sample agent observations, meaning that in our simulation an individual agent $i$'s prediction takes the form of a predicted profile $\sigma^*_i$, from which we can easily compute the belief profile $\beta_i$. Note that this is not required by our formal statement of belief profiles in \textsc{PeerBTS}, but it is a useful simplification for empirical simulation. Under this model we define four simple agent population models. \textbf{Clairvoyant} agents with perfect predictive accuracy; \textbf{Random} agents who predict random values; \textbf{Divided} where the top $50\%$ of agents are Clairvoyant and the bottom $50\%$ Random; and \textbf{Deciles} where agents are divided in deciles with decreasing predictive accuracy. Additional details regarding these models can be found in \Cref{apx:prediction_model}.

\subsection{Evaluation Metrics}

Previous work on peer selection has mostly evaluated mechanism performance via the proportion agents selected from an assumed ground truth ranking, i.e., recall at $k$ \cite{boehmer2024guide}. This is a useful metric, but is insufficient to interpret the impact of different prediction models. To fill this gap we define the mean lottery share for agent $i$ as the average probability share $i$ was awarded across all predictions in their review assignment, formally $\bar{\ell_i} = \nicefrac{1}{m} \sum_{j \in A(i)} L_j(i)$. This allows us to quantify an agents overall performance in the RBTS lottery and check that the mechanism is incentivizing effortful predictions by rewarding skilled predictors.

\subsection{Results}

\paragraph{Mean Lottery Share.} Examining the measure of mean lottery share, i.e., amount of reward for predicting well, across all experimental values shows a small but positive impact when agents are better predictors. In all cases the mean lottery share consistently discriminates between agents based on prediction quality, however the absolute differences are small. An example of this dynamic in the Deciles population model can be seen in \Cref{fig:lottery_share_graph} with additional results in \Cref{apx:full_lottery}. Selecting for parameter values which produce close to the maximal difference the gap in lottery share between highly accurate predictors and random or low skill predictors is $\approx 0.2$, meaning about a $20\%$ increase in reward. While this is not as strong as we might like it does demonstrate the proven theoretical incentive for agents to improve their predictions.

\paragraph{Mechanism Recall}

The recall at $k$ of \textsc{PeerBTS} is not a primary concern of our analysis. However, it is still useful to understand how \textsc{PeerBTS} performs, especially as the share of selections assigned to the RBTS lottery (reward) increases. This analysis is heavily dependent on the assumption that higher quality agents are superior predictors. The impact on recall caused by selecting from the RBTS lottery is entirely determined by the relationship between agent quality and predictive effort/skill. For this reason we only examine recall under the Deciles prediction model, as this model most directly connects agent rank and effort.

As can be seen in \Cref{fig:recall_graph}, the overall mechanism performance declines as additional selections are allocated to to the RBTS (reward) lottery. In particular, the prediction lottery selects more agents outside of the true $\{1, ..., k\}$, leading to lower recall as $d$ increases. This is consistent with the earlier measurements of lottery share under the Decile model (\Cref{fig:lottery_share_graph}), as the relatively low differences in selection probability between agents means that \textsc{PeerBTS} has a limited ability to select agents in the recall set if the link between effort and quality is not sufficiently strong.  We additionally tested a deterministic variant of \textsc{PeerBTS} based on \textsc{Partition}\cite{DBLP:conf/tark/AlonFPT11} which removes agents with a below average RBTS score. We found this mechanism to have inferior recall and lottery share performance, full details in \Cref{apx:det_baseline}.

One advantage of \textsc{PeerBTS} is that it's performance does not greatly degrade as the overall noise $\Phi$ increases. In contrast, \textsc{PeerNomination} (and by extension all other traditional peer selection mechanisms) struggles as agent reviews diverge more dramatically from the ground truth. Despite this advantage, our results suggest that in most cases the prediction lottery share $d$ should be kept relatively small relative to $k$. That is, a few slots should be reserved for the reward, but most selections should be made based on quality.

\subsection{Results Discussion}

Our simulations reveal distinct practical limitations in \textsc{PeerBTS}, especially in regards to consistently rewarding/selecting skilled predictors. Even large differences in agent predictive accuracy, such as those in the Divided and Decile models, had very small differences in the mean lottery share given to the best and worst predictors at $.1041$ and $.3226$ respectively. While this has no bearing on the theoretical properties of \textsc{PeerBTS}, it suggests that any actual implementation of the hybrid mechanisms we propose should be done cautiously and with the knowledge that any agents selected through the RBTS lottery are not necessarily of particularly high quality in the ground truth. The most obvious step to take to mitigate this concern is to limit the number of selections made using the prediction lottery, as only one selection is technically required to create our desired incentive for high quality predictions.

\subsection{Real World Data}

As part of the review process for a smaller (53 submissions, 60 reviewers) workshop we ran a pilot study to examine the quality of reviewer predictions.  
Reviewers were asked to predict the average final review score that papers would receive (6 per paper), and these predictions along with the final review scores were collected.\footnote{We were given access to this data under CC BY 4.0 and present only a limited analysis of this data as it pertains to \textsc{PeerBTS}, a full report and mass release of the data is forthcoming from the workshop organizers.} 
During the bidding process reviewers were asked to produce a brief micro-review of $6$ randomly assigned papers, focused on an initial assessment of fit for the workshop and a prediction of final scores for the full reviews. We found that even in this limited setting, reviewers were fairly accurate predictors, with a Mean Absolute Difference of $.69$ on a $5$-point scale between predictions and aggregate final review scores with $89.5\%$ of reviewers having an average MAD less than $1$. We also saw meaningful variation between individual reviewer's prediction abilities (SDev $.43$) which can be used to identify exceptional predictors worthy of a reward - for instance $33\%$ of reviewers had an average MAD below $.25$ and would be an excellent group for any mechanism to reward. 
This limited experiment suggest that two of the foundational assumptions underlying \textsc{PeerBTS} are feasible: that reviewers can, with effort, make accurate predictions and that meaningful variation between reviewer's predictive accuracy exist for mechanism designers to reward and punish.

\section{Conclusions and Future Work}
We have presented a novel peer selection mechanism -- \textsc{PeerBTS} -- that both incentivizes agents to report their true beliefs, as well as incentivizing them to put in effort to ``discover'' their predictions and beliefs. We show that to do so, the mechanism designer must require agents to report more than just agent beliefs on their peers, and we suggest a prediction algorithm based on RBTS \cite{robustBTS} for which the additional information is an estimate of others' beliefs, side-by-side by one's own review. This prediction algorithm can be combined with \emph{any} other peer-selection mechanism, resulting in a reward structure, such that the combined mechanism keeps the truth-incentivizing properties of its component mechanisms, while rewarding agents which have put in an effort.

We believe we are the first to suggest such a combination of both truth and effort incentives, and we hope it will result in further initiatives in this domain. Our mechanism in merely incentive compatible, but we hope that a strategyproof mechanism is out there as well
. As our empirical evaluations show, there is plenty more to do: first and foremost, trying to find an exact variant (or show it is impossible), and secondly, constructing reward mechanisms that more sharply differentiate between better and worse reviewers. Additionally, it is fascinating to see what can be done with further information requests beyond beliefs, and beliefs about others' beliefs.


\begin{acks}
Mattei was supported in part by NSF Awards, IIS-RI-2134857, IIS-RI-2339880 and CNS-SCC-2427237 as well as the Harold L. and Heather E. Jurist Center of Excellence for Artificial Intelligence at Tulane University and the Tulane University Center of Excellence for Community-Engaged Artificial Intelligence (CEAI). Portions of this research were conducted with high performance computational resources provided by the Louisiana Optical Network Infrastructure (LONI) (http://www.loni.org).
\end{acks}


\bibliographystyle{ACM-Reference-Format} 
\bibliography{peerbib}


\clearpage
\appendix

\section{Appendix}

\subsection{Proofs for Conditions for Incentive Compatability}\label{apx:IC_Proof}

For the lottery mechanism described in \Cref{alg:lottery-mechanism} to inherit the strict Bayes-Nash incentive compatibility of the RBTS, RBTS' conditions need to be met: All agents involved need to have a shared admissible prior and there are at least $3$ agents involved in any given round of the lottery, which is equivalent to requiring $m\geq 3$. Of these, the admissible prior is the only difficulty, as we have not previously defined the beliefs of agents involved in the peer selection in probabilistic terms. 

The requirements for an admissible prior are detailed in \citet{robustBTS}, Definition 1, but for our purposes we need to define a set of states $T$ with $|T| \geq 2$ ($T$ will be, for us, the state of nature -- any specific instantiation of agents' beliefs) and common priors $p_{j,t} = P(\apr[j][i] = 1 | t)$ such that:

\begin{enumerate}
 \item$\forall j, k \in \N, t \in T:\, p_{j,t} = p_{k,t}$.
 \item $\forall j \in \N, t \in T:\, 0 < p_{j,t} < 1$.
 \item $\forall j \in \N\, \exists t_1, t_2 \in T:\, p_{j, t_1} \neq p_{j, t_2}$
\end{enumerate}

We will show that assuming agent beliefs are sampled from the same distribution fulfill these requirements. Such a distribution can be a noisy variant of a ground truth, such as a Mallows' distribution \cite{mallows1957non} (the common theoretical assumption made by peer-selection papers), or any other relevant distribution to the particular situation.

To prove \Cref{thm:bnic_proof}, we first need several lemmata.
\begin{lemma}
\label{lem:distribution}
If all agent beliefs (i.e., their preferences over others) are drawn from shared distribution $D \in \Delta(\N^{\N})$ (i.e., $\forall j \in \N: \,\sigma_j \sim D$), then there exists a shared prior $\forall i, j \in \N, d \in D(\N^{\N}):\, P(\apr[j][i] = 1 | D)$.
\end{lemma}
\begin{proof}
 Recall that $\apr[i][j] = \begin{cases}
 1 & \text{if } \sigma_j(i) \leq q\\
 0 & \text{otherwise.} 
 \end{cases}$. Thus, it is fully determined by $\sigma_j$ -- $\apr[i][j]$ inherits the distribution from $D$ and we have a shared prior over $D$.
\end{proof}

As this result suggests, we will define states $T$ as a set of distributions over beliefs. To satisfy the requirement that the prior be fully mixed we take $T \subseteq \Delta^+(\N^\N)$ ($\Delta^{+}(X)$ being items with full support over the simplex of elements of $X$).

\begin{lemma}
\label{lem:state_priors}
If $T \subseteq \Delta^+(\N^\N)$ and $t_1, t_2 \in T$ exist such that $\forall j \in \N:\, p_{j, t_1} \neq p_{j, t_2}$, then any distribution over $T$ that supports $t_1, t_2$ forms an admissible prior. 
\end{lemma}
\begin{proof}
\Cref{lem:distribution} shows that any shared distribution over preferences forms a shared prior for agent approval, satisfying condition $1$. We can see that condition $2$ since elements in $T$ have a positive probability for any element in $\N^\N$ (yet only $k$ elements are selected from $\N$), each element $i\in\N$ has a positive probability of not being selected and of being selected, i.e., $ 0 < p_{j,t} < 1$, satisfying condition $2$. Finally, condition $3$ is satisfied as it is explicitly assumed by the lemma.
\end{proof}
 
 We can now finish proving \Cref{thm:bnic_proof}:
 \begin{proof}[Proof of \Cref{thm:bnic_proof}]
      \Cref{lem:distribution,lem:state_priors} show that our use of the RBTS algorithm, maintains its Bayes-Nash incentive compatible property. So what is left is to show that our combination of the RBTS with the regular mechanism (\Cref{alg:lottery-mechanism}) does so as well. The probability of selecting agent $j$ in the final lottery over $\mathcal{B}$ is proportional to the number of $j$ entries in $\mathcal{B}$.
 Therefore, $j$ maximizes their chance of being selected by maximizing the number of entries they receive in $\bigbag$, i.e., by maximizing the number of expected wins of lotteries $L_i$ where $j \in A^{-1}(i)$. As each of these lotteries are independent, we can consider $j$'s goal as maximizing $L_i(j)$ for some arbitrary $i$. Since $L_i(j) = \frac{\text{Scores}_i(j)^\epsilon}{2^\epsilon \cdot m}$, and both $\epsilon$ and $m$ are set in advance, this means agent $j$ seeks to maximize their entry in $\text{Scores}_i$, i.e., their RBTS score.

 We know that for agent populations $\geq 3$ and under admissible priors that the RBTS scoring mechanism is Bayes-Nash incentive compatible \cite[Theorem 9]{robustBTS}, and therefore that our RBTS lottery mechanism is as well.
 \end{proof}

\begin{corollary}\label{cor:IC}
\textsc{PeerBTS} incentivizes agents to report $\sigma_{i}^{*}$.
\end{corollary}

\begin{proof}[Proof of \Cref{cor:IC}]
Agent $i$'s chance of being selected in $S_{f}$ is not influenced by its $\sigma_{i}$, thanks to $f$'s strategyproofness. Agent $i$'s score in $H$ is also not improved by changing $\sigma_{i}^{*}$. So no agent would intentionally report $\sigma'_{i}\neq\sigma_{i}^{*}$, when $\sigma_{i}^{*}$ is known

Reporting a good $\prd[i]$ is strictly beneficial for the score in $H$. Thus, producing a good estimate of others, requires putting in the effort required to produce a good $\sigma_{i}^{*}$ (as we noted in the comment following \Cref{2bitsCor}).
\end{proof}

\subsection{Simulation Parameters} \label{apx:params}

Simulations were run with constant $n = 120, k = 20$ and varying $m \in \{3, 4, 6, 9\}, d \in \{2, 5, 10\}, e \in \{1, 2, 4, 8\},\varphi \in \{0.0, 0.2, 0.5\}$ (review board size, RBTS lottery selection size, lottery exponentiation term, and dispersion parameter for agent profiles respectively). 

\subsection{Agent Prediction Model Details} \label{apx:prediction_model}
In general, we can define an agent's predicted profile with the single dispersion parameter $\varphi^*_i$ as $$\sigma^*_i(j) = \begin{cases}
 \sigma_i & j = i \\
 \sigma^*_i \sim M_{\varphi^*_i}(\sigma_j) & j \neq i
\end{cases} \forall j \in \N$$
Where $M_{\varphi^*_i}(\sigma_j)$ is the Mallows distribution on permutations of $\sigma_j$ under the dispersion parameter $\varphi^*_i \in [0,1]$ as discussed previously. This is not a complete model of all ways agents could make predictions, but it is a flexible framework for studying the impact of prediction quality on agent outcomes under \textsc{PeerBTS}. While any set of $\N$ $\varphi^*$ values are sufficient to define the prediction skills of an agent population, we restricted our simulations to four simplified population models,

\begin{itemize}
 \item Clairvoyant: $\forall i \in \N: \varphi^*_i = 0.0 $, representing a population of agents with perfect predictive accuracy.
 \item Random: $\forall i \in \N: \varphi^*_i = 1.0 $, representing a population of agents who make predictions uniformly at random.
 \item Divided: $\varphi^*_i = \begin{cases}
 0.0 & i < \frac{|N|}{2} \\ 1.0 & i \geq \frac{|\N|}{2}
 \end{cases}$, representing an agent population where the top half of agents $1, .., \frac{|N|}{2}$ are perfect predictors and the bottom half of agents are random predictors.
 \item Deciles: $\varphi^*_i = D(i) \in \{0.1, 0.2, \ldots, 1.0\}$, representing a division of the agent population into discrete deciles where agent ranking determines predictive accuracy.
\end{itemize}

The two uniform population models (Clairvoyant and Random) are proof of concept population concepts intended to create scenarios where there is no connection between agent quality and prediction ability, while the two mixed population models reflect an assumption that high quality agents will be high quality predictors. We do not claim either that these are the only potentially interesting population models, or that they reflect a realistic population of agents, but they allow us to perform a basic performance evaluation and ensure that \textsc{PeerBTS} functions well in stylized context where underlying agent quality is directly correlated with their predictive ability.

\subsection{Expanded Lottery Share Figure} \label{apx:full_lottery}

This is an extended version of \Cref{fig:lottery_share_graph} which includes all tested agent prediction models.

\begin{figure*}[ht]
 \centering
 \includegraphics[width=0.95\textwidth]{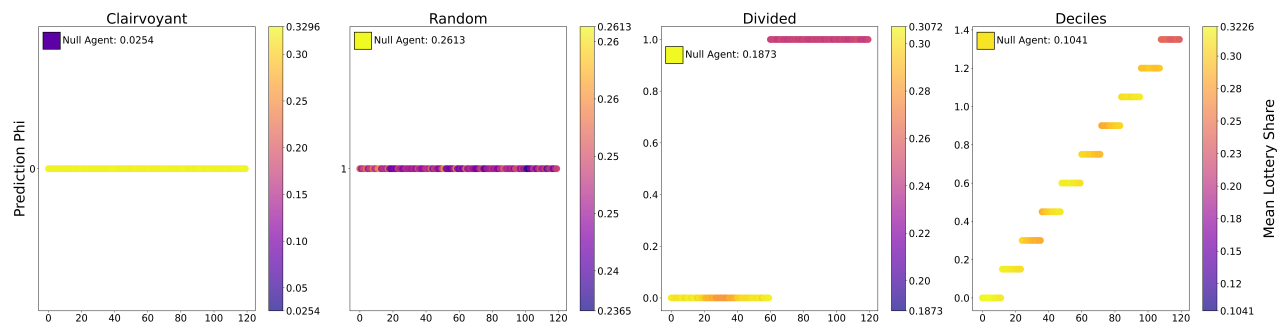}
 \caption{Lottery share for all population models at $m = 3$, $e = 8.0$, and $\varphi = 0.5$.}
 \label{fig:full_lottery}
 \Description{Complete lottery share graphic over all tested agent prediction models.}
\end{figure*}

\subsection{Deterministic Baseline}\label{apx:det_baseline}

We ran a deterministic version of \textsc{PeerBTS} as a baseline test. This does require significant modification to the selection method as we can no longer ensure incentive compatibility using a weighted lottery. Instead we establish some threshold $t$ and compute the final selection as $W/{j \in \N | \overline{\text{Scores}_i(j)} < t \forall i, j \in \N}$, that is to say we eliminate all agents from the selection who's average RBTS score is below $t$. To instantiate this mechanism we used \textsc{Partition}\cite{DBLP:conf/tark/AlonFPT11} as it is a common deterministic strategyproof mechanism. We selected $t$ to be the average RBTS score across all agents, meaning that we removed any agent selected by \textsc{Partition} who was a below average predictor. This resulted in notably diminished performance as can be seen in \cref{fig:det_recall}.

\begin{figure*}[ht]
 \centering
 \includegraphics[width=0.7\textwidth]{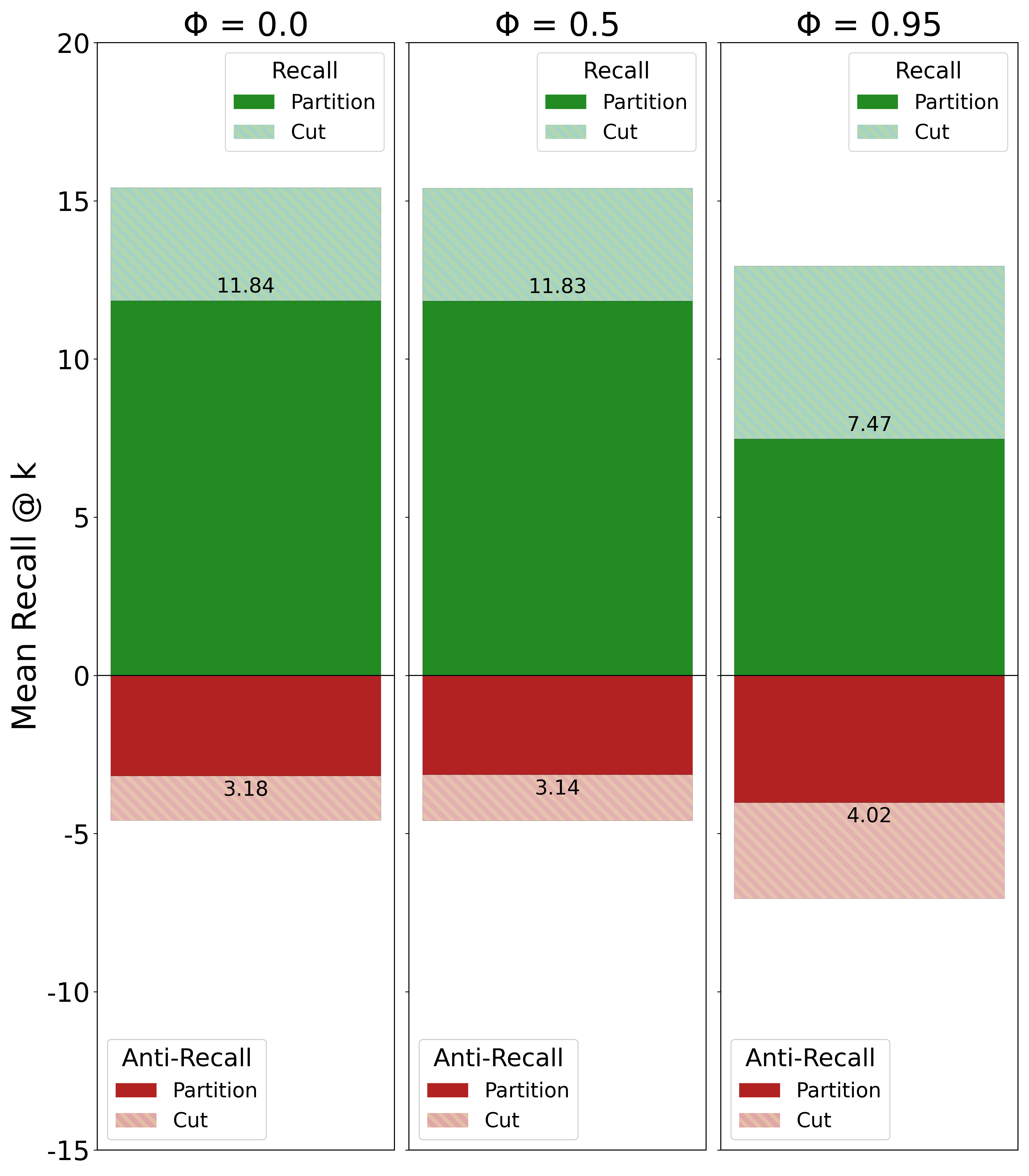}
 \caption{Mean recall and anti-recall of ground truth for deterministic \textsc{PeerBTS} using \textsc{Partition} and elimintating any agents scoring below population average on the RBTS.}
 \label{fig:det_recall}
\end{figure*}

\end{document}